\documentclass[12pt,reqno]{amsart}

\usepackage{euscript}
\usepackage{cite}
\usepackage[T1]{fontenc}
\usepackage[english]{babel}
\usepackage[all]{xy}
\usepackage{graphicx}
\usepackage{verbatim,amssymb}

\newtheorem{thm}{Theorem}[section]

\theoremstyle{definition}
\newtheorem{defn}[thm]{Definition}
\theoremstyle{remark}
\newtheorem{rem}[thm]{Remark}

\DeclareMathOperator*{\PLUS}{\oplus}
\allowdisplaybreaks
\begin{document}
\allowdisplaybreaks
\title[Trace formulae for Schrodinger operators on graphs]{Trace formulae for Schrodinger operators on metric graphs with applications to recovering matching conditions}
\author{Yulia Ershova}
\address{Institute of Mathematics, National Academy of Sciences of Ukraine. 01601 Ukraine, Kiev-4,
3, Tereschenkivska st.}
\email{julija.ershova@gmail.com}
\author{Alexander V. Kiselev}
\address{Department of Higher Mathematics and Mathematical Physics,
St. Petersburg State University, 1 Ulianovskaya Street,
St. Petersburg, St. Peterhoff 198504 Russia}
\email{alexander.v.kiselev@gmail.com}
\thanks{The second authors' work was partially supported by the RFBR, grant no. 12-01-00215-a.}

\begin{abstract}
The paper is a continuation of the study started in \cite{Yorzh1}.
Schrodinger operators on finite compact metric graphs are considered under the assumption that the matching conditions at
the graph vertices are of  $\delta$ type.
Either an infinite series of trace formulae (provided that edge potentials are infinitely smooth) or a finite number of such formulae
(in the cases of $L_1$ and $C^M$ edge potentials) are obtained
which link together two different quantum graphs under the assumption that their spectra coincide.
Applications are given to the problem of recovering matching conditions for a quantum graph based on its spectrum.
\end{abstract}

 \dedicatory{To Professor Yu.~S. Samoilenko on the
occasion of his 70th birthday}

\maketitle
\specialsection{Introduction}
In the present paper we focus our attention on the so-called quantum graph, i.e., a metric graph $\Gamma$ coupled with an associated second-order differential operator
acting on the Hilbert space $L^2(\Gamma)$ of square summable functions on the graph with an additional assumption that the functions belonging to
the domain of the operator are coupled by certain matching conditions at the graph vertices. These matching conditions reflect the graph
connectivity and usually are assumed to guarantee self-adjointness of the operator. Recently these operators have attracted a considerable interest
of both physicists and mathematicians due to a number of important physical applications, e.g., to the study of quantum wavequides. Extensive literature
on the subject is surveyed in, e.g., \cite{Kuchment}.

The present paper is devoted to the study of an inverse spectral problem for Schrodinger operators on finite compact metric graphs.
One might classify the possible inverse problems on quantum graphs in the following way.
\begin{itemize}
\item[(i)] Given spectral data, edge potentials and the matching conditions (usually one assumes standard matching conditions, see below), to reconstruct the metric graph;
\item[(ii)] Given the metric graph, edge potentials and spectral data, to reconstruct the matching conditions.
\item[(iii)] Given the metric graph, the spectral data and the matching conditions, to reconstruct the edge potentials.
\end{itemize}
There exists an extensive literature devoted to the problem
\emph{(i)}. To name just a few, we would like to mention the
pioneering works \cite{Roth,Smil1,Smil2} and later contributions
\cite{Kura1,Kura2,Kura3,Kostrykin}. These papers utilize an
approach to the problem \emph{(i)} based on the so-called trace
formula which relates the spectrum of the quantum graph to the set
of closed paths on the underlying metric graph. Different
approaches to the same problem were developed, e.g., in
\cite{pivo,Belishev_tree,Belishev_cycles}. The problem
\emph{(iii)} is the generalization of the classical inverse
problem for Sturm-Liouville operators and thus unsurprisingly has
attracted by far the most interest. We don't plan to dwell on this
any further as it is far beyond the scope of the present paper.

On the other hand, the problem \emph{(ii)} has to the best of our
knowledge surprisingly attracted much less interest. We believe it
was first treated in \cite{Carlson}. In the cited paper the square
of self-adjoint operator of the first derivative was treated (thus
not allowing for either $\delta-$ or $\delta'-$ coupling at the
graph vertices) on a subset of metric graphs (although cyclic
graphs were allowed). Then, after being mentioned in \cite{Kura2},
it was treated in \cite{Avdonin}, but only in the case of star
graphs. Then in our paper \cite{Yorzh1} we suggested an approach
based on the theory of boundary triplets which allowed us to
derive an infinite series of so-called trace formulae for graph
Laplacians.

The present paper is devoted to the analysis of the same problem
\emph{(ii)} in an attempt to generalize results of \cite{Yorzh1}
to the general setting of quantum graphs with summable edge
potentials. Unlike \cite{Avdonin}, we consider the case of a
general connected compact finite metric graph (in particular, this
graph is allowed to possess cycles), but only for two classes of
matching conditions at the graph vertices, namely,
 in either the case of $\delta$ type
matching conditions at the vertices or the case of $\delta'$ type matching conditions (see Section 2 for definitions). The approach utilized
is the same as in our work \cite{Yorzh1}, but the results obtained unsurprisingly look and feel much more involved. The named two classes singled out by us
prove to be physically viable \cite{Exner1, Exner2}.

In contrast to \cite{Avdonin}, where the spectral data used in order to reconstruct the matching conditions it taken to be the Weyl-Titchmarsh M-function (or
Dirichlet-to-Neumann map) of the graph boundary, we use the spectrum of the Schrodinger operator on a graph (counting multiplicities) as the data known to us from the outset.

The approach suggested is based on the celebrated theory of boundary triples \cite{Gor}. The concept of a generalized Weyl-Titchmarsh M-function for a properly chosen
maximal (adjoint to a symmetric, which we refer to as \emph{minimal}) operator allows us to reduce the study of the spectrum of the Schrodinger
operator on a  metric graph to the study of ``zeroes''
of the corresponding finite-dimensional analytic matrix function.
In order to achieve this goal, we surely have to construct an M-function for the whole graph rather than consider the Dirichlet-to-Neumann map pertaining
to the graph boundary. In sharp contrast to the situation of graph Laplacians where we were able to come up with an explicit formula for the M-function, in the
situation of Schrodinger operators we are only able to derive its asymptotic expansion (or rather, the first few terms of the latter if no smoothness
is required of edge potentials). Nevertheless, this limited information still allows us to derive a (finite) series of trace formulae
which link together two different Schrodinger operators on the same graph provided that their spectra coincide. These trace formulae surprisingly only involve the (diagonal) matrices
of coupling constants (i.e., constants appearing in matching conditions) and the asymptotics of diagonal entries of the Weyl-Titchmarsh M-function
of the graph $\Gamma$.

Moreover, the number of trace formulae available to us is limited by the smoothness of edge potentials, i.e., the more derivatives of the latter can be taken, the
more trace formulae come to existence.

The paper is organized as follows.

Section 2 introduces the notation and contains a brief summary of the material on the boundary triples used by us in the sequel. We continue by providing
an explicit  asymptotic expansion of the Weyl-Titchmarsh M-function written down in what we would like to think of as its ``natural'' form
 for the case of $\delta$ type matching conditions.

Section 3 contains our main result, i.e., the trace formulae for quantum graphs with $\delta$ type  matching conditions. As a corollary, we are able to prove that if
all the coupling constants on a quantum graph are identical, then the spectrum of the operator uniquely determines this universal coupling constant.

As for the case of $\delta'$ type matching conditions, we formulate the analogous result without any further discussion as the proof of it can be easily obtained along the same lines.

We would also like to mention that in contrast to our work \cite{Yorzh1}, in the present paper we do not allow the underlying metric graph to possess loops. Although
the machinery developed by us allows for the consideration of graphs with loops and the
corresponding results will have almost the same form, we have refrained from considering them here in view of keeping the paper transparent and
as easily readable as possible.

\specialsection{Boundary triples approach}

\subsubsection*{Definition of the Schrodinger operator on a quantum graph}

In order to define the quantum graph, i.e., the Schrodinger operator on a metric graph, we begin with the following

\begin{defn}
We call $\Gamma=\Gamma(\mathbf{E_\Gamma},\sigma)$ a finite compact metric graph, if it is a collection of a finite non-empty set
$\mathbf{E_\Gamma}$ of finite closed intervals  $\Delta_j=[x_{2j-1},x_{2j}]$,
$j=1,2,\ldots, n$, called \emph{edges}, and of a partition
$\sigma$ of the set of endpoints $\{x_k\}_{k=1}^{2n}$ into $N$ classes, $\mathbf{V_\Gamma}=\bigcup^N_{m=1} V_m$. The equivalence classes
 $V_m$, $m=1,2,\ldots,N$ will be called \emph{vertices} and the number of elements belonging to the set $V_m$ will be called the \emph{valence} of the vertex
$V_m$.
\end{defn}

With a finite compact metric graph $\Gamma$ we associate the Hilbert space
$$L_2(\Gamma)=\PLUS_{j=1}^n L_2(\Delta_j).$$
This Hilbert space obviously doesn't feel the connectivity of the graph, being the same for each graph with the same number of edges of the same lengths.

In what follows, we single out two natural \cite{Exner1} classes of so-called \emph{matching conditions} which lead to a properly defined self-adjoint operator
on the graph $\Gamma$, namely, the matching conditions of $\delta$ and $\delta'$ types. In order to describe these, we will introduce the following notation.
For a smooth enough function $f\in L_2(\Gamma)$, we will use throughout the following definition of the normal derivative on a finite compact metric graph:
$$\partial_n f(x_j)=\left\{ \begin{array}{ll} f'(x_j),&\mbox{ if } x_j \mbox{ is the left endpoint of the edge},\\
-f'(x_j),&\mbox{ if } x_j \mbox{ is the right endpoint of the edge.}
\end{array}\right.$$

\begin{defn}\label{def_matching} If $f\in \PLUS_{j=1}^n W_2^2 (\Delta_j)$ and $\alpha_m$ is a complex number (referred to below as a coupling constant),
\begin{itemize}
\item[\textbf{($\delta$)}] the condition of continuity of the function $f$
through the vertex $V_m$ (i.e., $f(x_j)=f(x_k)$ if $x_j,x_k\in V_m$) together with the condition
$$
\sum_{x_j \in V_m} \partial _n f(x_j)=\alpha_m f(V_m)
$$
is called $\delta$-type matching at the vertex $V_m$;
\item[\textbf{($\delta'$)}] the condition of continuity of the normal derivative $\partial_n f$
through the vertex $V_m$ (i.e., $\partial_n f(x_j)=\partial_n f(x_k)$ if $x_j,x_k\in V_m$) together with the condition
$$
\sum_{x_j \in V_m} f(x_j)=\alpha_m \partial_n f(V_m)
$$
is called $\delta'$-type matching at the vertex $V_m$;
\end{itemize}
\end{defn}

\begin{rem}
Note that the $\delta$-type matching condition in a particular case when $\alpha_m=0$ reduces to the so-called standard, or Kirchhoff, matching condition
at the vertex $V_m$.
Note also that at the graph boundary (i.e., at the set of vertices of valence equal to 1) the $\delta$- and $\delta'$-type conditions reduce to the usual 3rd type
ones, whereas the standard matching conditions lead to the Neumann condition at the graph boundary.
\end{rem}

We are all set now to define the  Schrodinger operator on the graph $\Gamma$ with $\delta$- or $\delta'$-type matching conditions.
\begin{defn}\label{def_Laplace}
The Schrodinger operator $A$ on a graph $\Gamma$ with $\delta$-type ($\delta'$-type, respectively) matching conditions is the operator defined by the differential
expression $-\frac{d^2}{dx^2}+q(x)$ on $L_2(\Gamma)$, where real-valued $q(x)|_{\Delta_j}\equiv q_j(x)\in L_1(\Delta_j)$ are referred to as \emph{edge potentials},
in the Hilbert space $L_2(\Gamma)$ on the domain of functions belonging to the Sobolev space $\PLUS_{j=1}^n W_2^2(\Delta_j)$ and satisfying $\delta$-type
($\delta'$-type, respectively) matching conditions at every vertex $V_m$, $m=1,2,\dots,N.$
\end{defn}

\begin{rem}
Note that the matching conditions reflect the graph connectivity: if two graphs with the same edges have different topology, the resulting operators are different.
\end{rem}

Provided that all coupling constants $\alpha_m$, $m=1\dots N$, are real, it is easy to verify that the operator $A$ is self-adjoint in the
Hilbert space $L_2(\Gamma)$ \cite{Exner1,KostrykinS}. Throughout the present paper, we are going to consider this self-adjoint situation only, although it has
to be noted that the approach
developed can be used for the purpose of analysis of the general non-self-adjoint situation as well (under the additional assumption that all edge potentials are
still real-valued).

Clearly, the self-adjoint operator thus defined on a finite compact metric graph has purely discrete spectrum that might accumulate to $+\infty$ only. In order
to ascertain this, one only has to note that the operator considered is a finite-dimensional perturbation in the resolvent sense of the direct sum of Sturm-Liouville
operators on the individual edges.

\begin{rem}
Note that w.l.o.g. each edge $\Delta_j$ of the graph $\Gamma$ can be considered to be an interval $[0,l_j]$, where $l_j=x_{2j}-x_{2j-1}$, $j=1\dots n$ is
the length of the corresponding edge. Indeed, performing the corresponding linear change of variable one reduces the general situation to the one where
all the operator properties depend on the lengths of the edges rather than on the actual edge endpoints. Throughout the present paper we will therefore only consider
this situation.
\end{rem}

\begin{rem}
Note that unlike the case of graph Laplacians, in general case of Schrodinger operators the underlying metric graph has to be thought of as \emph{oriented}, i.e., each
edge $\Delta_i$ has a starting point and an endpoint ($0$ and $l_i$, respectively, following our convention).
\end{rem}

In order to treat the inverse spectral problem (ii) for graph Schrodinger operators, we will first need to get some in-depth information on
 the generalized Weyl-Titchmarsh M-function of the operator considered. The most elegant and straightforward way to do so is in our view
by utilizing the apparatus of boundary triples developed in \cite{Gor,Ko1,Koch,DM}. We briefly recall the results
essential for our work.

\subsubsection*{Boundary triplets and the Weyl-Titchmarsh matrix M-function}

Suppose that $A_{min}$ is a symmetric densely defined closed
linear operator acting in the Hilbert space $H$ ($D(A_{min})\equiv
D_{A_{min}}$ and $R(A_{min})\equiv R_{A_{min}}$ denoting its domain and
range respectively; $D(A_{max})\equiv D_{A_{max}}$,
$R(A_{max})\equiv R_{A_{max}}$ denoting the domain and range of
operator $A_{max}$ adjoint to $A_{min}$). Assume that $A_{min}$ is completely nonselfadjoint (simple)\footnote{It is easy to see
that all the results of the present paper still hold even if the underlying minimal operator is not simple. Nevertheless, the problem of its simplicity
is of an independent interest, and we refer the reader to our paper \cite{Yorzh1} where this question was studied for graph Laplacians. The situation of Schodinger
operators can be analyzed along the same lines.},
i.e., there exists no reducing subspace $H_0$ in $H$ such that the
restriction $A_{min}|H_0$ is a selfadjoint operator in $H_0.$ Further
assume that the deficiency indices of $A_{min}$ (probably
being infinite) are equal: $n_+(A_{min})=n_-(A_{min})\le\infty.$

\begin{defn}[\cite{Gor,Ko1,DM}]\label{BT_def}
Let $\Gamma_0,\ \Gamma_1$ be linear mappings of $D_{A_{max}}$ to
$\mathcal{H}$ -- a separable Hilbert space. The triple $(\mathcal{H},
\Gamma_0,\Gamma_1)$ is called \emph{a boundary triple}
for the operator $A_{max}$ if:
\begin{enumerate}
\item for all $f,g\in D_{A_{max}}$
$$
(A_{max} f,g)_H -(f, A_{max}
g)_H = (\Gamma_1 f, \Gamma_0 g)_{\mathcal{H}}-(\Gamma_0 f,
\Gamma_1 g)_{\mathcal{H}}.
$$
\item the mapping $\gamma$ defined as $f\longmapsto (\Gamma_0 f; \Gamma_1
f),$ $f\in D_{A_{max}}$ is surjective, i.e.,
for all $Y_0,Y_1\in\mathcal{H}$ there exists such $y\in
D_{A_{max}}$ that $\Gamma_0 y=Y_0,\ \Gamma_1 y =Y_1.$
\end{enumerate}
\end{defn}

A boundary triple can be constructed for any operator $A_{min}$ of the class considered. Moreover, the
space $\mathcal H$ can be chosen in a way such that $\dim \mathcal H=n_+=n_-.$

\begin{defn}[\cite{Gor,DM}]
A nontrivial extension ${A}_B$ of the operator $A_{min}$ such that $A_{min}\subset  A_B\subset A_{max}$  is
called \emph{almost solvable} if there exists a boundary triple $(\mathcal{H},
\Gamma_0,\Gamma_1)$ for $A_{max}$ and a bounded linear operator $B$ defined everywhere on
$\mathcal{H}$ such that for every $f\in D_{A_{max}}$
$$
f\in D_{A_B}\text{ if and only if } \Gamma_1 f=B\Gamma_0 f.
$$
\end{defn}

It can be shown that if an extension $A_B$ of $A_{min}$, $A_{min}\subset  A_B\subset A_{max}$, has regular
points (i.e., the points belonging to the resolvent set)
in both upper and lower half-planes of the complex plane, then this extension is almost solvable.

The following theorem holds:
\begin{thm}[\cite{Gor,DM}]\label{old-extra}
Let $A_{min}$ be a closed densely defined symmetric operator with $n_+(A_{min})=n_-(A_{min}),$
let $(\mathcal{H},\Gamma_0,\Gamma_1)$ be a boundary triple of $A_{max}$.
Consider the almost solvable extension
$A_B$ of $A_{min}$ corresponding to the bounded operator $B$
in $\mathcal{H}.$ Then:
\begin{enumerate}
\item $y\in D_{A_{min}}$ if and only if $\Gamma_0 y=\Gamma_1 y=0,$
\item $ A_B$ is maximal, i.e., $\rho( A_B)\not=\emptyset$,
\item $(A_B)^*\subset A_{max},\ (A_B)^*=
A_{B^*},$
\item operator $A_B$ is dissipative if and only if $B$
is dissipative,
\item $(A_B)^*=A_B$ if and only if $B^*=B.$
\end{enumerate}
\end{thm}

The generalized Weyl-Titchmarsh M-function is then defined as follows.
\begin{defn}[\cite{DM,Gor,Koch}]\label{M-def}
Let $A_{min}$ be a closed densely defined symmetric operator,
$n_+(A_{min})=n_-(A_{min}),$ $(\mathcal{H},\Gamma_0,\Gamma_1)$ is its
space of boundary values. The operator-function $M(\lambda),$
defined by
\begin{equation}\label{Weyleq}
M(\lambda)\Gamma_0 f_{\lambda}=\Gamma_1 f_{\lambda},
\ f_{\lambda}\in \ker (A_{max}-\lambda),\  \lambda\in
\mathbb{C}_\pm,
\end{equation}
is called the Weyl-Titchmarsh M-function of a symmetric operator $A_{min}.$
\end{defn}

The following Theorem describing the properties of the M-function clarifies its meaning.
\begin{thm}[\cite{Gor,DM}, in the form adopted in \cite{RyzhovOTAA}]\label{Weyl}
Let $M(\lambda)$ be the M-function of a symmetric operator
$A_{min}$ with equal deficiency indices ($n_+(A_{min})=n_-(A_{min})<\infty$).
Let $A_B$ be an almost solvable extension of
$A_{min}$ corresponding to a bounded operator $B.$ Then for every $\lambda\in
\mathbb{C}:$
\begin{enumerate}
\item $M(\lambda)$ is analytic operator-function when
$Im\lambda\not=0$, its values being bounded linear operators
in $\mathcal{H}.$
\item $(Im\ M(\lambda))Im\ \lambda>0$ when $Im \lambda\not =0.$
\item $M(\lambda)^*=M(\overline{\lambda})$ when $Im \lambda\not =0.$
\item $\lambda_0\in \rho(A_B)$ if and only if $(B-M(\lambda))^{-1}$ admits bounded analytic continuation into the point $\lambda_0$.
\end{enumerate}
\end{thm}

In view of the last Theorem, one is tempted to reduce the study of the spectral properties of the Schrodinger operator on a metric graph to the study
of the corresponding Weyl-Titchmarsh M-function. Indeed, if one considers the operator under investigation as an extension of a properly chosen
symmetric operator defined on the same graph and constructs a boundary triple for the latter, one might utilize all the might of the complex analysis and the
theory of analytic matrix R-functions, since in this new setting the (pure point) spectrum of the quantum Laplacian is located exactly at the points into which
the matrix-function $(B-M(\lambda))^{-1}$ cannot be extended analytically (wagely speaking, these are ``zeroes'' of the named matrix-function).

It might appear as if the non-uniqueness of the space of boundary values and the resulting non-uniqueness of the Weyl-Titchmarsh M-function leads to
some problems on this path; but on the contrary, this flexibility of the apparatus is an advantage of the theory rather than its weakness. Indeed, as we are
going to show below, this allows us to ``separate'' the data describing the metric graph (this information will be carried by the M-function) from the data
describing the matching conditions at the vertices (this bit of information will be taken care of by the matrix $B$ parameterizing the extension). In turn,
this ``separation'' proves to be quite fruitful in view of applications that we have in mind.

Following \cite{Yorzh1}, we proceed with an explicit construction of the ``natural'' boundary triple for quantum graphs.

\subsubsection*{Construction of a boundary triple and asymptotics of the Weyl-Titchmarsh M-function for quantum graphs}
Let $\Gamma$ be a fixed finite compact metric graph. Let us denote by $\partial\Gamma$ the graph boundary, i.e., all the vertices of the graph
which have valence 1. We further assume that at all the vertices the matching conditions are of $\delta$ type.

As the operator $A_{max}$ rather then $A_{min}$ is crucial from the point of view of construction of a boundary triple, we start with this maximal operator and explicitly
describe its action and domain: $A_{max}= - \frac{d^2}{dx^2}+q(x)$,
\begin{equation}\label{DAmax}
D(A_{max})=\left\{ f
\in \bigoplus_{j=1}^n W^2_2(\Delta_j)\ |\ \forall\ V_m \in V_{\Gamma
\backslash
\partial
\Gamma} f \text{ is continuous at }
V_m \right\}.
\end{equation}
\begin{rem}
Note that the operator chosen is not the ``most maximal'' maximal one: one could of course skip the condition of continuity through internal vertices;
nevertheless, the choice made proves to be the most natural from the point of view expressed above. This is exactly due to the fact that the graph connectivity is thus
reflected in the domain of the maximal operator and therefore propels itself into the expression for the M-matrix. Moreover, it should be noted that
this choice is also natural since the dimension of the M-matrix will be exactly equal to the number of graph vertices.
\end{rem}

The choice of the operators $\Gamma_0$ and $\Gamma_1$, acting onto
$\mathbb{C}_N$, $N = |V_\Gamma|$ is made as follows (cf., e.g., \cite{Aharonov} where a similar choice is suggested, but only for the graph boundary):
\begin{equation}\label{BT_formula}\Gamma_0f=\left( \begin{array}{c} f(V_1)\\f(V_2)\\ \ldots\\
f(V_N)
\end{array}\right);\quad
\Gamma_1f=\left( \begin{array}{c} \sum\limits_{x_j:x_j \in V_1} \partial_n f(x_j)\\\sum\limits_{x_j:x_j \in V_2} \partial_n f(x_j)\\ \ldots\\
\sum\limits_{x_j:x_j \in V_N} \partial_n f(x_j)
\end{array}\right).\end{equation}
Here the symbol $f(V_j)$ denotes the value of the function $f(x)$ at the vertex $V_j$. The latter is meaningful because of the choice of the domain of the maximal operator.

One can prove \cite{Yorzh1}, that the triple $(\mathbb{C}_N; \Gamma_0, \Gamma_1)$, $N =
|V_\Gamma|$ is a boundary triple for the operator $A_{max}$ in the sense of Definition \ref{BT_def}.

\begin{rem}
If one considers a graph Laplacian with matching conditions of $\delta'$ type, the choice of the maximal operator and the corresponding
boundary triple  has to change accordingly:
$A_{max}= - \frac{d^2}{dx^2}+q(x)$,
\begin{gather}
D(A_{max})=\left\{ f
\in \bigoplus_{j=1}^{n} W^2_2(\Delta_j) | \forall\ V_m \in
V_{\Gamma \backslash
\partial
\Gamma}\ \partial_n f \text{ is continuous at } V_m \right\}\label{DAmaxp}\\
\Gamma_0 f=\left( \begin{array}{c} \partial_n f(V_1)\\ \partial_n f(V_2)\\ \ldots\\
\partial_n f(V_N)
\end{array}\right);\quad
\Gamma_1f=-\left( \begin{array}{c} \sum\limits_{x_j:x_j \in V_1} f(x_j)\\\sum\limits_{x_j:x_j \in V_2} f(x_j)\\ \ldots\\
\sum\limits_{x_j:x_j \in V_N} f(x_j)
\end{array}\right).\label{BT_formula_p}
\end{gather}
\end{rem}

In order to establish the asymptotic behavior of the
Weyl-Titchmarsh M-matrix in the case of Schrodinger operator on a metric graph,
we first consider the following auxiliary problem on the interval $[0,l]$ (here $l$ will be any of the edge lengths of the graph $\Gamma$.
We have elected to drop the lower index altogether to simplify the notation in the hope that this will not lead to any ambiguity).
\begin{equation} -y'' + q(x)y=k^2y \label{aux_problem}
\end{equation} with $\delta$ type conditions at both ends, which in our setting amounts to
$y'(0)=\beta_1y(0)$; $-y'(l)=\beta_2y(l)$.

The problem of finding the $M$-matrix function asymptotics as $\lambda\to-\infty$ for a given graph $\Gamma$ with
matching conditions of $\delta$ type reduces to finding asymptotics of solutions to our auxiliary problem with boundary conditions $y(0)=0; $
$y(l)=1$ and $y(0)=1;$ $y(l)=0$, respectively. Indeed, given our choice of boundary triple for $\Gamma$, it is sufficient to consider only solutions
$u_j\in Ker(A_{max}-\lambda)$, $j=1,\dots,n$ such that $\Gamma_0 u_j$ is the all-zero vector but for 1 in the j-th position. Then the j-th column of $M(\lambda)$ will
coincide with the vector $\Gamma_1 u_j$ and henceforth, the asymptotics of the corresponding matrix elements will be given by the corresponding asymptotics of
$\Gamma_1 u_j$.

On the other hand, the vector $u_j$ defined above due to the choice of $\Gamma_0$ is nothing but a set of solutions of our auxiliary problem on each edge containing
the vertex $V_j$ (in fact, the solution to the auxiliary problem \eqref{aux_problem} with boundary conditions $y(0)=1;$ $y(l)=0$ if $V_j$ is the left endpoint of
the named edge and with $y(0)=0; $
$y(l)=1$ in the opposite case), whereas on all the other edges it is clearly identically zero if $k^2$ is not in the spectrum of the auxiliary problem on the corresponding
edge with Dirichlet boundary conditions, which is clearly the case when the potential is summable and $k^2$ is sufficiently large negative number.

As we are going to assume $\lambda\to -\infty$, it proves worthwhile to put $\sqrt{\lambda}\equiv k=i\tau$, $\tau
\rightarrow +\infty$. This convention will be used throughout.

Denote by $\phi(x,k)$ the solution to the equation \ref{aux_problem} with Cauchy data $\psi(0,k)=0$, $\psi'(0,k)=1$ (the so-called sine-type
solution). This standard solution clearly exists and satisfies the integral equation \cite{Sargsyan}
\begin{equation} \psi(x,k)=\frac{\sin{kx}}{k}+\frac{1}{k}\int_0^x
\sin{(k(x-t))q(t)\psi(t,k)dt}.\label{int_eq_aux}\end{equation}
We will then put $f_1(x,k)=\frac{\psi(x,k)}{\psi(l,k)}$ for the solution $f_1$ with boundary conditions $f_1(0,k)=0$ and $f_1(l,k)=1$, provided that we are not on
the spectrum of the corresponding Dirichlet problem which can be safely assumed to be granted.

Now assume $\tau$ large enough. Then the only condition that $q$ is summable over the interval leads to the standard way of obtaining the full asymptotical expansion of
$\psi(x,i\tau)$ in $\tau$ based on the first term which is known \cite{Sargsyan} to be $\psi(x,i\tau)=\frac{e^{\tau x}}{2\tau}+O\left(\frac{e^{\tau
x}}{\tau^2}\right)$. For the sake of completeness we briefly recall the corresponding details.

The integral equation \eqref{int_eq_aux} assumes the form
\begin{multline} \psi(x,i\tau)=\frac{e^{\tau
x}}{2\tau}-\\
\frac{e^{-\tau x}}{2\tau}+\frac{e^{\tau
x}}{2\tau}\int_0^x e^{-\tau t}q(t)\psi(t,i\tau)dt-\frac{e^{-\tau
x}}{2\tau}\int_0^x e^{\tau
t}q(t)\psi(t,i\tau)dt=\\
\frac{e^{\tau
x}}{2\tau}+\frac{e^{\tau x}}{2\tau}Y_1-\frac{e^{\tau x}}{2\tau}Y_2
+ o(e^{-\tau x}),
\label{int_eq_1}\end{multline} where $Y_1=\int_0^x e^{-\tau
t}q(t)\psi(t,i\tau)dt$, $Y_2=\int_0^x e^{\tau
(t-2x)}q(t)\psi(t,i\tau)dt$.
Taking the first term asymptotics for $\psi(x,i\tau)$ into account, one immediately obtains:
\begin{gather*}
\frac{e^{\tau x}}{2\tau}Y_1=\frac{e^{\tau x}}{(2\tau)^2} Q(x)+O\left(\frac{e^{\tau x}}{\tau^3}\right),\\
\frac{e^{\tau x}}{2\tau}Y_2=o\left(\frac{e^{\tau x}}{\tau^2}\right),
\end{gather*}
where $Q(x)= \int\limits_0^x q(t) dt$ and
since clearly
$$
Y_2=\frac1{2\tau}\int_0^x e^{2\tau
(t-x)}q(t)dt+O\left(\frac1{\tau^2}\right)\int_0^x e^{2\tau
(t-x)}q(t)dt=o(1/\tau).
$$
Therefore, $$\psi(x,i\tau)=\frac{e^{\tau x}}{2\tau}+\frac{e^{\tau
x}}{(2\tau)^2} Q(x)+o\left(\frac{e^{\tau x}}{\tau^2}\right).$$

Substituting this asymptotic formula into the integral equation \eqref{int_eq_1} again, one now obtains by induction:
\begin{equation}\label{asym_psi}
\psi(x,i\tau)=e^{\tau x}\sum_{j=1}^n \frac{Q^{j-1}(x)}{(j-1)!}\frac 1{{(2\tau)}^j}+ o\left(\frac{e^{\tau x}}{\tau^{n}}\right),
\end{equation}
where the following explicit calculation has been used:
\begin{multline*}
\int_0^x e^{-\tau t}q(t)\left[e^{\tau t}\sum_{j=1}^{n-1} \frac{Q^{j-1}(t)}{(j-1)!}\frac 1{{(2\tau)}^j}+ o\left(\frac{e^{\tau t}}{\tau^{n-1}}\right)\right]dt=\\
\sum_{j=1}^{n-1}\frac1{(j-1)!}\frac 1{{(2\tau)}^j}\int_0^x Q^{j-1}(t) dQ(t)+\int_0^x q(t)
o\left(\frac1{\tau^{n-1}}\right)dt=\\
\sum_{j=1}^{n-1}\frac1{j!}\frac 1{{(2\tau)}^j}Q^j(x)+o\left(\frac1{\tau^{n-1}}\right).
\end{multline*}

Our next task is to compute the corresponding asymptotic expansion for $\psi'(x,k)$. Together with \eqref{asym_psi} this will yield the full asymptotic expansion
in $\tau$ of the solution $f_1(x,i\tau)$ which is sufficient for our purposes (see the definition of $\Gamma_1$ \eqref{BT_formula}).

Differentiating \eqref{int_eq_aux}, one obtains for the named derivative:
\begin{multline}\label{int_eq_2}
\psi'_x(x,i\tau)=cos kx+\int\limits_0^x cos k
(x-t)q(t)\psi(t,k)dt=\\
\frac{e^{\tau
x}}{2}+\frac{e^{\tau x}}{2}\int_0^x e^{-\tau
t}q(t)\psi(t,i\tau)dt+\frac{e^{-\tau x}}{2}\int_0^x e^{\tau
t}q(t)\psi(t,i\tau)dt+O(e^{-\tau x}).
\end{multline}
Putting the asymptotic expansion \eqref{asym_psi} into \eqref{int_eq_2}, we easily obtain from the first two terms:
$$
\frac{e^{\tau
x}}{2}+\frac{e^{\tau x}}{2}\int_0^x e^{-\tau
t}q(t)\psi(t,i\tau)dt=\frac{e^{\tau x}}{2}\sum_{j=0}^n \frac{Q^{j}(x)}{(j)!}\frac 1{{(2\tau)}^j}+ o\left(\frac{e^{\tau x}}{\tau^{n}}\right).
$$
The third term proves to be harder to deal with. What's more, obtaining the full asymptotic expansion of it appears to be only possible under additional
smoothness restrictions imposed on the potential. We proceed as follows.
\begin{multline*}
\frac{e^{-\tau x}}{2}\int_0^x e^{\tau
t}q(t)\left[e^{\tau t}\sum_{j=1}^n \frac{Q^{j-1}(t)}{(j-1)!}\frac
1{{(2\tau)}^j}+ o\left(\frac{e^{\tau
t}}{\tau^{n}}\right)\right]dt =\\
\frac{e^{-\tau x}}{2}\left[\sum_{j=1}^n I_j \frac 1{{(2\tau)}^j}\right] +   o\left(\frac 1{\tau^{n}}\right),
\end{multline*}
where $I_j:=\frac 1{(j-1)!}\int_0^x e^{2\tau
t}q(t)Q^{j-1}(t)dt$.

Denote $Q_j(x)=\frac d{dx}\frac 1j Q^{j}(x)$, then using $Q'(x)=q(x)$ we obtain
$$I_j=\frac 1{(j-1)!}\int_0^x e^{2\tau
t}Q_j(t)dt.$$

If $q\in C^{n-1}([0,l])$, for every $j=1,\dots,n$ and $n\geq 2$ multiple integration by parts yields:
$$I_j=\sum_{m=0}^{n-j-1}e^{2\tau x}\frac{(-1)^m Q_j^{(m)}}{(2\tau)^{m+1}}+o\left(\frac{e^{2\tau x}}{\tau^{n-j}}\right).$$
Thus,
\begin{multline*}
\frac{e^{-\tau x}}{2}\left[\sum_{j=1}^n I_j \frac 1{{(2\tau)}^j}\right]=
\sum_{j=1}^n \frac{e^{\tau x}}{2(j-1)!}\sum_{m=0}^{n-j-1}\frac{(-1)^m Q_j^{(m)(x)}}{(2\tau)^{m+j+1}}=\\
\frac{e^{\tau x}}2 \sum_{k=0}^n \frac 1{(2\tau)^k}\sum_{j=1}^{k-1}\frac{(-1)^{k-j-1}}{(j-1)!}Q_j^{(k-j-1)}(x) + o\left(\frac{e^{\tau x}}{\tau^n}\right).
\end{multline*}
Note, that in the last sum there are actually no terms for $k=0,1$; we have elected to start the summation from $k=0$ for notational convenience reasons.

Now one finally gets the following asymptotic expansion for $\psi'(x,i\tau)$:
\begin{multline}\label{asym_psi_prime}
\psi'(x,i\tau)=\frac 12 e^{\tau x}\sum_{k=0}^n \frac 1{(2\tau)^k}\left[
\frac 1 {k!} Q^k(x)+\sum_{j=1}^{k-1} \frac {(-1)^{k-j-1}}{(j-1)!}Q_j^{(k-j-1)}(x)
\right]+\\
o\left(\frac{e^{\tau x}}{\tau^n}\right),
\end{multline}
where, as before, $Q(x)=\int_0^x q(t)dt$, $Q_j(x)=\frac 1j \frac d{dx} Q^j(x)$ and in particular, $Q_1(x)\equiv q(x)$.

We now turn our attention to finding the asymptotic expansion in $\tau$ of the solution $f_2(x,k)$ of \eqref{aux_problem} satisfying boundary conditions
$f(0,k)=1$, $f(l,k)=0$. Again, this solution is nothing but $f_2(x,k)=\frac{\phi(x,k)}{\phi(0,k)}$, where $\phi$ is the solution of Cauchy problem
with the following data: $\phi(l,k)=0$, $\phi'(l,k)=1$. For $\phi$ one has the following integral equation, similar to \eqref{int_eq_aux}
$$\phi(x,k)=\frac{\sin
(k(x-l))}{k}-\frac1k\int_x^l \sin (k(x-t))q(t)\phi(t,i\tau)dt.$$

Proceeding analogously to the treatment of the solution $\psi(x,k)$, it is not hard to obtain the following asymptotic expansion of the solution $\phi$:
\begin{equation}\label{asym_phi}
\phi(x,i\tau)=-e^{\tau(l-x)}\sum_{j=1}^n \frac{R^{j-1}(x)}{(j-1)!}\frac 1{(2\tau)^j}+o\left(\frac {e^{\tau (l-x)}}{\tau^n}\right).
\end{equation}
Here the folowing notation has been adopted: $R(x):=\int_x^l q(t)dt$.

Moreover, a similar analysis leads to the following asymptotics for the function $\phi'(x,i\tau)$:
\begin{multline}\label{asym_phi_prime}
\phi'(x,i\tau)=\frac 12 e^{\tau (l-x)}\sum_{k=0}^n \frac 1{(2\tau)^k}\left[
\frac 1 {k!} R^k(x)+\sum_{j=1}^{k-1} \frac {1}{(j-1)!}R_j^{(k-j-1)}(x)
\right]+\\
o\left(\frac{e^{\tau (l-x)}}{\tau^n}\right),
\end{multline}
where $R_j(x):=-\frac 1j \frac d{dx} R^j(x)$ and in particular, $R_1(x)\equiv q(x)$.

The analysis carried out above leads to the following
\begin{thm}\label{M-delta}
Let $\Gamma$ be a finite compact metric graph with no loops. Let the operator $A_{max}$ be the Shrodinger operator on the domain \eqref{DAmax} with potentials
$\{q_j\}_{j=1}^n$, $n$ being the number of edges of the graph $\Gamma$. For all $j$ let $q_j\in C^{M}(\Delta_j)$ for some natural $M$. Let
the boundary triple for $A_{max}$ be chosen as $(\mathbb{C}^{N}, \Gamma_0, \Gamma_1)$, where $N$ is the number of vertices of $\Gamma$ and the operators
$\Gamma_0$ and $\Gamma_1$ are defined by \eqref{BT_formula}. Then the generalized Weyl-Titchmarsh M-function is an $N\times N$ matrix with matrix
elements admitting the following asymptotic expansions as $\lambda\to -\infty$.
\begin{equation}
m_{jp}=\begin{cases}
{\scriptstyle \tau \sum\limits_{\Delta_t \in E_j}\sum_{j=0}^{M+1} \frac {b_{j}^{(t)}}{\tau^j} +
   \tau \sum\limits_{\Delta_t \in E'_j}\sum_{j=0}^{M+1} \frac {a_{j}^{(t)}}{\tau^j} + o(\frac 1{\tau^M}),}& {\scriptstyle j=p}, \\
{\scriptstyle o(\frac 1{\tau^{\tilde M}}) \text{ for all } \tilde M>0,}& {\scriptstyle j\neq p, \text{{\scriptsize vertices }} V_j \text{{\scriptsize and }} V_p}\\
& \text{{\scriptsize are connected by an edge}},\\
{\scriptstyle 0},& {\scriptstyle j\neq p,\, \text{{\scriptsize
vertices}}\, V_j\, \mbox{{\scriptsize and}}\, V_p}
\\
&\text{{\scriptsize are not connected by an edge}}.
\end{cases}
\end{equation}
Here $k=\sqrt{\lambda}$ (the branch of the square root is fixed so that $\text{Im } k\geq 0$), $\tau=-ik$;
$E_j$ is the set of graph edges
such that their left endpoints belong to the vertex
$V_j$; $E'_j$ is the set of graph edges
such that their right endpoints belong to the vertex $V_j$; and finally,
$\{a_{j}^{(t)}\}_{j=0,\dots,M+1}$ and $\{b_{j}^{(t)}\}_{j=0,\dots,M+1}$ are two sets of real numbers which are uniquely determined by
the potentials on the edges belonging to $E'_j$ and $E_j$, respectively.
\end{thm}

\begin{proof}
The statement follows immediately from asymptotic expansions \eqref{asym_phi},\eqref{asym_phi_prime},\eqref{asym_psi} and \eqref{asym_psi_prime}. Indeed, consider the
vertex $V_j$. Then the $j$-th column of the matrix $M(\lambda)$ is given by $\Gamma_1 u_j$, where $u_j\in Ker(A_{max}-\lambda)$ is such that $\Gamma_0 u_j=
(0,\dots,1,0,\dots,0)$ with 1 in the $j$-th position. As $\lambda$ is assumed to be sufficiently large negative, $u_j$ is identically zero on every edge not belonging
to either $E_j$ or $E'_j$. Then, the matrix elements $m_{ij}$ are identically zero for such vertices $V_i$. Next, if vertices $V_i$ and $V_j$ are connected by an edge,
the corresponding matrix element is either of the form $f'_1(0,k)=\psi'(0,k)/\psi(l,k)\equiv 1/\psi(l,k)$ or of the form
$f'_2(l,k)=\phi'(l,k)/\phi(0,k)\equiv 1/\phi(0,k)$ (depending on the direction of the edge conneting
$V_i$ and $V_j$) and thus decays faster than $1/\tau^M$ for any positive $M$ (see \eqref{asym_phi}, \eqref{asym_psi}). Finally,
the diagonal element $m_{jj}$ is the sum of terms of the form $-f'_1(l,k)=-\psi'(l,k)/\psi(l,k)$ over all edges belonging to $E'_j$ and of terms of the
form $f'_2(0,k)=\phi'(0,k)/\phi(0,k)$ over all edges belonging to $E_j$, from where the claim follows almost immediately.
\end{proof}

We remark that all the constants $\{a_{j}^{(t)}\}_{j=0,\dots,M+1}$ and $\{b_{j}^{(t)}\}_{j=0,\dots,M+1}$ for $t:\ \Delta_t\in E_j$ and $t:\ \Delta_t\in E'_j$, respectively,
appearing in the statement of the latter Theorem, can be computed explicitely based on the asymptotic expansions \eqref{asym_phi},\eqref{asym_phi_prime},\eqref{asym_psi} and \eqref{asym_psi_prime}.
We have elected not to include the corresponding rather trivial but lengthy calculations in the present paper in view of its better readability.

We would also like to point out that essentially the just proven Theorem reads: ``The more smooth the edge potentials are, the more terms of the asymptotic
expansion of the Weyl-Titchmarsh matrix-function one gets''. Comparing this situation with what one faces in the case of graph Laplacian, where one always gets a full
asymptotic expansion, one ends up with the question: ``Is this how the world is made, or is our method of proof deficient?''. In fact, the paper \cite{Savchuk} suggests
that given edge potentials in the class $C^M$ one should be able to get at least one more term in the asymtotics, compared to our answer. However, A.A. Shkalikov \cite{Shkalikov}
comes to the conclusion that this is the end of the story, i.e., a full asymptotic expansion is only possible for infinitely smooth edge potentials. This question
in fact turns to be crucial, see Section 3 below.

Theorem \ref{M-delta} can be made more transparent in two special cases which in view of what follows are of high importance to us. First, we assume that
mean values of all potentials $\{q_j\}_{j=1}^n$ are zero. Then the asymptotic expansions obtained above simplify drastically.

\begin{thm}\label{zeromean}
Let $\Gamma$ be a finite compact metric graph with no loops. Let the operator $A_{max}$ be the Shrodinger operator on the domain \eqref{DAmax} with potentials
$\{q_j\}_{j=1}^n$, $n$ being the number of edges of the graph $\Gamma$. For all $j$ let $q_j\in C^{M}(\Delta_j)$ for some natural $M$. Assume further that the mean values
of all the potentials $q_j$ are zero. Let
the boundary triple for $A_{max}$ be chosen as in Theorem \ref{M-delta}. Then the generalized Weyl-Titchmarsh M-function is an $N\times N$ matrix with matrix
elements admitting the following asymptotic expansions as $\lambda\to -\infty$.
\begin{equation}\label{M-delta2}
m_{jp}=\begin{cases}
{\scriptstyle \sum\limits_{\Delta_t \in E_j}\left(-\tau -\tau\sum_{k=2}^{M+1} \frac 1{(2\tau)^k}\sum_{j=1}^{k-1} \frac 1{(j-1)!} B_{jk}^t\right) +}&\\
    {\scriptstyle+\sum\limits_{\Delta_t \in E'_j}\left(-\tau -\tau\sum_{k=2}^{M+1} \frac 1{(2\tau)^k}\sum_{j=1}^{k-1} \frac {(-1)^{k-j-1}}{(j-1)!} A_{jk}^t\right) +}&\\
    {\scriptstyle+ o(\frac 1{\tau^M}),}& {\scriptstyle j=p}, \\
{\scriptstyle o(\frac 1{\tau^{\tilde M}}) \text{ for all } \tilde M>0,}& {\scriptstyle j\neq p, \text{{\scriptsize vertices }} V_j \text{{\scriptsize and }} V_p}\\
& \text{{\scriptsize are connected by an edge}},\\
{\scriptstyle 0},& {\scriptstyle j\neq p,\, \text{{\scriptsize
vertices}}\, V_j\, \mbox{{\scriptsize and}}\, V_p}
\\
&\text{{\scriptsize are not connected by an edge}}.
\end{cases}
\end{equation}
Here $k=\sqrt{\lambda}$ (the branch of the square root is fixed so that $\text{Im } k\geq 0$), $\tau=-ik$;
$E_j$ is the set of the graph edges
such that their left endpoints belongs to the vertex
$V_j$; $E'_j$ is the set of the graph edges
such that their right endpoints belongs to the vertex $V_j$. Finally,
\begin{gather*}
A_{jk}^t=\left.\frac 1j \frac {d^{k-j}}{dx^{k-j}}\left(\int_0^x q_t(y) dy\right)^j\right|_{x=l_t}\\
B_{jk}^t=-\left.\frac 1j \frac {d^{k-j}}{dx^{k-j}}\left(\int_x^{l_t} q_t(y) dy\right)^j\right|_{x=0}
\end{gather*}
\end{thm}

\begin{proof}
Again, the statement follows immediately from asymptotic expansions \eqref{asym_phi},\eqref{asym_phi_prime},\eqref{asym_psi} and \eqref{asym_psi_prime}. One only has to take
into account that the condition of zero means for the potentials $q_j$ implies that $Q(l)=R(0)=0$ for all edges belonging to $E'_j$ and $E_j$, respectively.
\end{proof}

We remark \label{emphasis} that under the assumptions of the latter Theorem the first two terms of the asymptotic expansion of $M(\lambda)$ as $\lambda\to-\infty$ turn out to be exactly
the same as in the situation of zero potentials, cf. \cite{Yorzh1}. What's even more revealing, one might prove that the picture remains unchanged even when
the edge potentials are no longer required to have zero means. Moreover, this fact holds in the most general case considered in the present paper,
i.e., in the case when the edge potentials are only assumed to be summable.

\begin{thm}\label{nonsmooth}
Let $\Gamma$ be a finite compact metric graph with no loops. Let the operator $A_{max}$ be the Shrodinger operator on the domain \eqref{DAmax} with potentials
$\{q_j\}_{j=1}^n$, $n$ being the number of edges of the graph $\Gamma$. For all $j$ let $q_j\in L_1(\Delta_j)$.  Let
the boundary triple for $A_{max}$ be chosen as in Theorem \ref{M-delta}. Then the generalized Weyl-Titchmarsh M-function is an $N\times N$ matrix with matrix
elements admitting the following asymptotic expansions as $\lambda\to -\infty$.
\begin{equation}\label{M-delta3}
m_{jp}=\begin{cases}
 {\scriptstyle -\gamma_j \tau + o(1),}& {\scriptstyle j=p}, \\
{\scriptstyle o(\frac 1{\tau^M}) \text { for any } M>0,}& {\scriptstyle j\neq p, \text{{\scriptsize vertices }} V_j \text{{\scriptsize and }} V_p}\\
& \text{{\scriptsize are connected by an edge}},\\
{\scriptstyle 0},& {\scriptstyle j\neq p,\, \text{{\scriptsize
vertices}}\, V_j\, \mbox{{\scriptsize and}}\, V_p}
\\
&\text{{\scriptsize are not connected by an edge}}.
\end{cases}
\end{equation}
Here $k=\sqrt{\lambda}$ (the branch of the square root is fixed so that $\text{Im } k\geq 0$), $\tau=-ik$;
finally, $\gamma_j$ is the valence of the vertex $E_j$.
\end{thm}

The \emph{proof} is again a straightforward computation based on asymptotic expansions \eqref{asym_phi},\eqref{asym_phi_prime},\eqref{asym_psi} and \eqref{asym_psi_prime}.
One has to notice that in order to get the first two terms of asymptotic expansions \eqref{asym_phi_prime} and \eqref{asym_psi_prime}, which are the only necessary terms
to write down the asymptotic formula for $M(\lambda)$ up to $o(1)$ as $\lambda\to -\infty$, all integrations by parts become redundant and thus there is no need
to assume that the edge potentials are smooth.

Notice also that Theorems \ref{zeromean} and \ref{nonsmooth} in fact state that the first \emph{two} terms of the asymptotic expansion of Titchmarsh-Weyl M-matrix
($O(\tau)$ and $O(1)$) do not ``feel'' the edge potentials as both are exactly the same as in the case of a graph Laplacian. This fact will be a cornerstone of the
analysis presented in the next Section.

\specialsection{Trace formulae for a pair of Schrodinger operators on the same metric graph}
In the present section, we apply the mathematical apparatus developed in Section 2 in order to study isospectral (i.e., having the same spectrum, counting
multiplicities) Schrodinger operators defined on a finite compact metric graph $\Gamma$. In order to do so, we will assume that the graph itself is given. Moreover,
we will assume that the matching conditions at all its vertices are of $\delta$ type ($\delta'$ type case can be treated analogously, and we formulate the corresponding
result towards the end of the Section).

\begin{thm}\label{trace}
Let $\Gamma$ be a finite compact metric graph with no loops having $N$ vertices. Let $A_{B_1}$ and $A_{B_2}$ be two Schrodinger operators
on the graph $\Gamma$ with $\delta$-type
matching conditions ($B_1=\text{diag}\{\tilde \alpha_1,\dots,\tilde \alpha_N\}$ and $B_2=\text{diag}\{ \alpha_1$, $\dots$, $\alpha_N\}$, where
both sets $\{\tilde \alpha_m\}$ and $\{\alpha_m\}$ are the sets of coupling constants in the sense of Definition \ref{def_matching}). Let the
edge potentials $q_j$ be the
same for both Schrodinger operators, and let further $q_j\in C^M(\Delta_j)$ for all $j=1,\dots,n$ for some integer constant $M \geq 0$ (for notational convenience,
we imply $C^0\equiv L^1$).
Let the (point) spectra of these two operators (counting multiplicities) be equal, $\sigma(A_{B_1}) = \sigma(A_{B_2})$.

Then the following formula holds for every $s=1,\dots,M+1$:
\begin{multline}\label{finite_set}
\sum_{i=1}^N\sum_{m=1}^{M+1}\sum_{\sum_{p=1}^{M+1}j_p=m; \sum_{p=1}^{M+1}p j_p=s}\frac {(-1)^{m-j_1}m!}{j_1!\cdots j_{M+1}!}\frac 1{\gamma_i^m}\tilde\alpha_i^{j_1}(\delta_1^{(i)})^{j_2}\cdots (\delta_M^{(i)})^{j_{M+1}}=\\
\sum_{i=1}^N\sum_{m=1}^{M+1}\sum_{\sum_{p=1}^{M+1}j_p=m; \sum_{p=1}^{M+1}p j_p=s}\frac {(-1)^{m-j_1}m!}{j_1!\cdots j_{M+1}!}\frac 1{\gamma_i^m}\alpha_i^{j_1}(\delta_1^{(i)})^{j_2}\cdots (\delta_M^{(i)})^{j_{M+1}}
\end{multline}
where $\gamma_i$, $i=1,\dots,N$ are valences of the vertices in the graph $\Gamma$; real constants $\delta_k^{(i)}$ are uniquely determined by the underlying
metric graph and the edge potentials.
\end{thm}

\begin{rem}
Note that the formulae \eqref{finite_set} are a system of non-linear equations linking together two sets of coupling constants, $\{\tilde\alpha_i\}$ and
$\{\alpha_i\}$. In general, we are unable to tell whether or not these for some $M$ large enough reduce to the uniqueness result, i.e., $\alpha_i=\tilde
\alpha_i$ $\forall i$. This appears generally not true even in the much simpler case of quantum Laplacian, although in the latter case the infinite series of
such formulae can be shown to almost yield uniqueness.
\end{rem}

\begin{proof}
We will use the apparatus developed in Section 2. Namely, we choose the maximal operator $A_{max}$ as in \eqref{DAmax}, the boundary triple \eqref{BT_formula} and
use the asymptotic expansion for the Weyl-Titchmarsh M-function of $A_{max}$ obtained in Theorem \ref{M-delta}. Then w.r.t. the chosen boundary triple
the operators $A_{B_1}$ and $A_{B_2}$ are both almost solvable extensions of the operator $A_{min}=A_{max}^*$, parameterized by the matrices
$B_1$ and $B_2$, respectively. Throughout we of course assume that $B_1-B_2\not=0$.

We will now show that provided that the spectra of both given operators coincide, $\det
(B_1-M(\lambda))(B_2-M(\lambda))^{-1}\equiv 1$. This is done by a Liouville-like argument in exactly the same way as in \cite{Yorzh1}. We first verify that
the named determinant is in fact a ratio of two scalar analytic entire functions $F_1$ and $F_2$.
Moreover, by Theorem \ref{Weyl} their fraction $F_1/F_2$ has no poles and no zeroes, since the spectra
of operators $A_{B_1}$ and $A_{B_2}$ coincide.

Now it can be easily ascertained that both $F_1$ and $F_2$ are of normal type and of order at least not greater than 1 \cite{Levin,Sargsyan}. Then
their fraction is again an entire function of order not greater than 1 \cite{Levin}. Finally, by Hadamard's theorem $\frac{F_1}{F_2}=e^{a\lambda+b}$.

It remains to be seen that $a=b=0$. This follows immediately from the asymptotic behaviour of the matrix-function $M(\lambda)$ as $\lambda\to-\infty$ (see
Theorem \ref{nonsmooth}).

We have thus obtained the following identity:
$$
1\equiv\det(B_1-M(\lambda))(B_2-M(\lambda))^{-1}
$$
On the other hand, from Theorems \ref{M-delta} and \ref{nonsmooth} it follows, that within the assumptions of the Theorem the diagonal entries of $M(\lambda)$ admit the
following asymptotic expansion as $\lambda\to-\infty$: for all $i=1,\dots,N$
\begin{equation*}
m_{ii}(\tau)=\gamma_i\tau+\frac {\delta_1^{(i)}}\tau+\dots+\frac{\delta_M^{(i)}}{\tau^M}+o(\frac 1{\tau^M}),
\end{equation*}
where all the coefficients $\delta_k^{(i)}$, $i=1,\dots,N$, $k=1,\dots,M$ are uniquely determined by the metric graph $\Gamma$ and edge potentials. As for non-diagonal
entries, these are irrelevant since they are either identically zero or decaying in $\tau$ faster than any inverse power.
Then
$$
\frac{\det(M(-\tau^2)-B_1)}{\det(M(-\tau^2)-B_2)}=\dfrac{\prod_{i=1}^n (1-\frac{\tilde\alpha_i}{\gamma_i\tau}+\frac{\delta_1^{(i)}}{\gamma_i\tau^2}+\cdots+\frac{\delta_M^{(i)}}{\gamma_i\tau^{M+1}})+o(\frac 1{\tau^{M+1}})}
{\prod_{i=1}^n (1-\frac{\alpha_i}{\gamma_i\tau}+\frac{\delta_1^{(i)}}{\gamma_i\tau^2}+\cdots+\frac{\delta_M^{(i)}}{\gamma_i\tau^{M+1}})+o(\frac 1{\tau^{M+1}})}
$$
Take logarithm of both sides of this identity, taking into account that the fraction of determinants on the left hand side is identically equal to 1. Thus,
all the coefficients of asymptotic expansion of the right hand side of the equation in inverse powers of $\tau$
have to be necessarily equal to zero. A straightforward calculation then completes the proof.

\end{proof}

The formulae \eqref{finite_set} seem to be quite involved; what's even worse, at the first sight they give impression of being formulated in a very implicit form.
In fact, this is not so since by Theorems of the preceding Section one knows how to calculate somewhat cryptic looking constants $\delta_k^{(i)}$ explicitly.
Nevertheless, the next Theorem demonstrates that these formulae are quite usable, at least in the situation when all the coupling constants at all vertices are assumed to be equal.

\begin{thm}
Let $\Gamma$ be a finite compact metric graph with no loops having $N$ vertices. Let $A_{\tilde\alpha}$ and $A_{\alpha}$ be two Schrodinger operators
on the graph $\Gamma$ with $\delta$-type
matching conditions with common for all vertices coupling constants $\tilde\alpha$ and $\alpha$, respectively.
Let all edge potentials\footnote{Note, that we do not need to assume here that the edge potentials are the same for operators
$A_{\tilde\alpha}$ and $A_{\alpha}$.} $q_i^{(1)},q_i^{(2)}\in L_1(\Delta_i)$ for all $i=1,\dots,N.$
Let the (point) spectra of these two operators (counting multiplicities) be equal, $\sigma(A_{B_1}) = \sigma(A_{B_2})$.
Then $\tilde\alpha = \alpha$.
\end{thm}

The \emph{proof} of the Theorem is a direct corollary of Theorem \ref{nonsmooth} and is essentially a simplified version of the proof of previous Theorem.
In the case considered one actually only needs the first formula of \eqref{finite_set} ($s=1$), which reads
$$
\sum_{i=1}^N \frac{\tilde \alpha}{\gamma_i} =\sum_{i=1}^N \frac{\alpha}{\gamma_i}
$$
(exactly the same as in the case of graph Laplacians) and quite obviously yields the claim.
The fact that
the edge potentials can be safely assumed to be different for the two operators considered follows from the fact emphasized on the page \pageref{emphasis},
that the first two terms of the asymptotic expansion of $M(\lambda)$ at minus infinity are absolutely the same as in the case of graph Laplacian, i.e.,
carry absolutely no information on the edge potentials.

For the sake of completeness, we formulate without a proof (which can be however quite easily obtained along the line of arguments presented in the present
paper, starting with the definition of a maximal operator \eqref{DAmaxp} and of an associated ``natural'' boundary triple \eqref{BT_formula_p}) the corresponding result
for the case of Schrodinger operators on metric graphs with $\delta'$ type matching conditions.

\begin{thm}
Let $\Gamma$ be a finite compact metric graph with no loops having $N$ vertices. Let $A_{\tilde\alpha}$ and $A_{\alpha}$ be two Schrodinger operators
on the graph $\Gamma$ with $\delta'$-type
matching conditions with common for all vertices coupling constants $\tilde\alpha$ and $\alpha$, respectively.
Let all edge potentials\footnote{As before, we do not need to assume here that the edge potentials are the same for operators
$A_{\tilde\alpha}$ and $A_{\alpha}$.} $q_i^{(1)},q_i^{(2)}\in L_1(\Delta_i)$ for all $i=1,\dots,N$.
Let the (point) spectra of these two operators (counting multiplicities) be equal, $\sigma(A_{B_1}) = \sigma(A_{B_2})$.
Then $\tilde\alpha = \alpha$.
\end{thm}

We have elected to postpone the presentation of examples to our
next paper devoted to the same subject due to the fact that
although the trace formulae derived above bring us very close to
proving uniqueness in the general case of different coupling
constants at the graph vertices, additional information in fact
needs to be taken into account in order to achieve this goal.
Namely, one needs to consider the poles of the matrix-function
$M(\lambda)-B$ and in particular the residues of its determinant.
The corresponding rather lengthy analysis will be presented
elsewhere, accompanied by some examples.

\subsection*{Acknowledgements} The authors express their deep gratitude to Prof. Sergey Naboko for his constant attention
to authors' work and for fruitful discussions. We would also like
to cordially thank Prof. Andrey Shkalikov (Moscow State Univ.) and
Dr. Vladimir Sloushch (St. Petersburg State Univ.) for their input
which was essential for us in the course of preparation of the
manuscript. We would also like to cordially thank our referees for
making some very helpful comments.

\end{document}